\documentclass[conference]{IEEEtran}
\usepackage[utf8]{inputenc}
\usepackage[draft]{hyperref}
\usepackage{enumitem}
\usepackage{amsthm,amssymb,amsmath}
\usepackage{graphicx,multirow}
\usepackage{hanging,colortbl}
\usepackage{listings,xcolor,sgame,float,bbm}
\usepackage{siunitx}
\usepackage{comment}
\usepackage{color,soul,mathtools}
\usepackage[noend]{algpseudocode}
\usepackage{algorithm,caption,subcaption}
\usepackage{stackengine}

\newtheorem{definition}{\textbf{Definition}
}
\newtheorem{theorem}{\textbf{Theorem}}
\newtheorem{corollary}{\textbf{Corollary}}[theorem]

\DeclareMathOperator*{\argmin}{\text{arg min}\,}

\DeclareMathOperator*{\sign}{\text{sign}\,}

\makeatletter
\newcommand{\ostar}{\mathbin{\mathpalette\make@circled\star}}
\newcommand{\make@circled}[2]{%
  \ooalign{$\m@th#1\smallbigcirc{#1}$\cr\hidewidth$\m@th#1#2$\hidewidth\cr}%
}
\newcommand{\smallbigcirc}[1]{%
  \vcenter{\hbox{\scalebox{0.77778}{$\m@th#1\bigcirc$}}}%
}
\makeatother

\algrenewcommand\algorithmicprocedure{}

\title{TASR: A Novel Trust-Aware Stackelberg Routing Algorithm to Mitigate Traffic Congestion \vspace{-0.5ex}}

\author{\IEEEauthorblockN{Doris E. M. Brown, Venkata Sriram Siddhardh Nadendla, and Sajal K. Das}
\vspace{0.5ex}
\IEEEauthorblockA{
Department of Computer Science 
\\ 
Missouri University of Science and Technology 
\\ 
Rolla, MO 65409, USA 
\\ Email: \{deby3f, nadendla, sdas\}@mst.edu}
\vspace{-4ex}
}

\begin{document}

\maketitle

\begin{abstract}
Stackelberg routing platforms (SRP) reduce congestion in one-shot traffic networks by proposing optimal route recommendations to selfish travelers. Traditionally, Stackelberg routing is cast as a partial control problem where a fraction of traveler flow complies with route recommendations, while the remaining respond as selfish travelers. 
In this paper, a novel Stackelberg routing framework is formulated where the agents exhibit \emph{probabilistic compliance} by accepting SRP's route recommendations with a \emph{trust} probability.
A greedy \emph{\textbf{T}rust-\textbf{A}ware \textbf{S}tackelberg \textbf{R}outing} algorithm (in short, TASR) is proposed for SRP to compute unique path recommendations to each traveler flow with a unique demand. Simulation experiments are designed with random travel demands with diverse trust values on real road networks such as Sioux Falls, Chicago Sketch, and Sydney networks for both single-commodity and multi-commodity flows. The performance of TASR is compared with state-of-the-art Stackelberg routing methods in terms of traffic congestion and trust dynamics over repeated interaction between the SRP and the travelers. Results show that TASR improves network congestion without causing a significant reduction in trust towards the SRP, when compared to most well-known Stackelberg routing strategies. 
\vspace{1ex}
\end{abstract}


\section{Introduction}
Modern navigation systems present quickest-route recommendations, as opposed to revealing the network state information, in order to alleviate the cognitive load at the travelers. However, such systems reinforce selfish routing amongst travelers, which can increase traffic congestion in transportation networks \cite{macfarlane2019, roughgarden2005, roughgarden2002bad}. Hence, the goal of a smart transportation system is to steer its travelers' routing equilibrium \cite{wardrop1952road} towards a social-optimum, be it in terms of mitigating traffic congestion, carbon emissions, or any other social-welfare metric on the entire transportation network. Stackelberg routing \cite{roughgarden2001stackelberg} is one such proposed solution that reduces traffic congestion using strategic route recommendations, without increasing travelers' cognitive load.

Traditionally, Stackelberg routing is cast as a partial control problem where a fraction of the traveler demand complies with route recommendations presented by the Stackelberg routing platform (SRP), and the rest of the demand chooses their paths selfishly \cite{korilis1997achieving}. However, in the real world, travelers exhibit probabilistic compliance based on their trust in the SRP, wherein the route recommendations are accepted with some probability. To this end, this paper attempts to redesign Stackelberg routing under stochastic compliance when the travelers' trust probability is known to the SRP, thereby allowing the SRP to better route traffic demands towards a network flow resulting in lower congestion. This solution approach is of particular importance in smart transportation systems in which solutions are needed to effectively persuade the drivers to choose network paths which minimize the overall network latency when the drivers would ordinarily behave selfishly, but may probabilistically comply with a system’s path recommendations under certain conditions.


\subsection{Related Work on Stackelberg Routing} 
Most of the literature on Stackelberg routing strategies 
focuses on theoretical guarantees on the price of anarchy in terms of traffic congestion 
under various assumptions such as arbitrary latency functions \cite{kaporis2006price}, \cite{swamy2012effectiveness}, arbitrary networks \cite{bonifaci2010stackelberg}, \cite{karakostas2009stackelberg}, horizontal queues \cite{krichene2014stackelberg}, \cite{krichene2017stackelberg}, and repeated interaction settings \cite{krichene2016social}. However, it is known that even for simple parallel-path networks with non-decreasing linear latency functions, computing the optimal Stackelberg routing strategy is NP-hard \cite{roughgarden2001stackelberg}. For this reason, approximate Stackelberg strategies were developed, including the Largest Latency First, Scale, Aloof \cite{roughgarden2001stackelberg}, and ASCALE \cite{bonifaci2007impact}. Each of these strategies are \textit{opt-restricted}, i.e., the traffic assignment on any given road segment due to the aforementioned strategies is upper-bounded by the \emph{optimal complaint flow} which minimizes the overall network congestion when all the travelers on the network are complaint.
Traditionally, in Stackelberg routing literature, the non-compliant demand is assumed to observe the path choices made by the compliant demand, and respond with path choices that form a user equilibrium. These strategies aim to reduce network congestion by restricting the route choices of the non-compliant demands to routes that would drive the resulting user equilibrium closer to that of the system optimum.

From a different perspective, efforts have been made to theoretically determine the minimum number of compliant agents present in a network in order for the system to minimize total network latency \cite{sharon2018traffic}. However, this work assumes travelers are distinctly compliant or noncompliant, with compliant travelers opting in to a micro tolling system, giving the system control over their route choices through charges for taking certain routes. Like the previously mentioned works on Stackelberg routing, this work does not assume that travelers may have any probabilistic compliance in the system, nor does it consider degradation in the trust of compliant travelers in the system as they repeatedly experience higher latency paths compared to those paths taken by selfish travelers, leading compliant travelers to eventually opt out.

While it is possible for a traffic control center to utilize a traditional Stackelberg routing strategy to interact with and share path recommendations to individual travelers via a GPS-enabled smartphone application or vehicle navigation system, such traditional Stackelberg strategies are impractical to implement in this way due to the inability of noncompliant travelers to know the path choices made by compliant travelers or to coordinate their behavior with other noncompliant travelers. Furthermore, in some Stackelberg routing settings, it is assumed that human travelers are noncompliant travelers while autonomous vehicles are compliant agents \cite{kolarich2022stackelberg}. However, it is disadvantageous to assume all human drivers in a traffic network are fully noncompliant, preventing some fraction of traffic flow from being routed towards the system optimum in cases where a human traveler would accept recommendations shared by the system.


\subsection{Novel Contributions} 
The main contributions of this paper are summarized below.
\vspace{-0.25ex}
\begin{enumerate}[label=(\roman*),leftmargin=3ex]
\setlength{\itemindent}{0ex}
\setlength{\itemsep}{1.5ex}
\item The interaction between the SRP and the individual travelers is modeled as a novel \textbf{\emph{probabilistic-compliance Stackelberg routing}} game, where each travel demand group (follower) probabilistically accepts a greedy route recommendation shared by the system (leader), and constructs a route decision. To the best of our knowledge, this is the first work to investigate probabilistic-compliance Stackelberg routing which is
apt to real-world travelers as they may or may not accept the recommendation depending on their trust in the system. 

\item A \textbf{\emph{novel greedy, trust-aware Stackelberg routing algorithm called TASR}} is developed to compute and share path-profile recommendations for groups of diverse travel demands to mitigate traffic congestion in a network. 

\item The proposed method is validated on diverse \textbf{\emph{realistic simulation experiments}} that simulate travelers (with varying trust levels) on the Sioux Falls (South Dakota), Chicago (Illinois) Sketch, and Sydney (Australia) networks \cite{stabler2020}, and is found to be superior over other well-known Stackelberg routing strategies such as Largest Latency First (LLF), Scale, and Aloof strategies found in \cite{roughgarden2001stackelberg}, and ASCALE found in \cite{bonifaci2010stackelberg} and  \cite{bonifaci2007impact}. 
\\
\end{enumerate}

The rest of the paper is organized as follows. In Section \ref{Sect: System Model}, the nonatomic Stackelberg routing game between a group of travelers and a central routing system is formally presented.  Section \ref{Sec: Proposed Methodology} presents the proposed trust-aware greedy Stackelberg routing strategy. Section \ref{Sect: Evaluation Methodology} describes the validation techniques and traffic network datasets utilized, while Section \ref{Sect: Results} discusses the experimental results. Section \ref{Sect: Conclusion} offers concluding remarks and plans for future work.

\section{System Model and Problem Formulation \label{Sect: System Model and Problem Formulation}}

\begin{figure}[!t]
\vspace{-0.1in}
\centering
\captionsetup{justification=centering}
\includegraphics[width=0.45\textwidth]{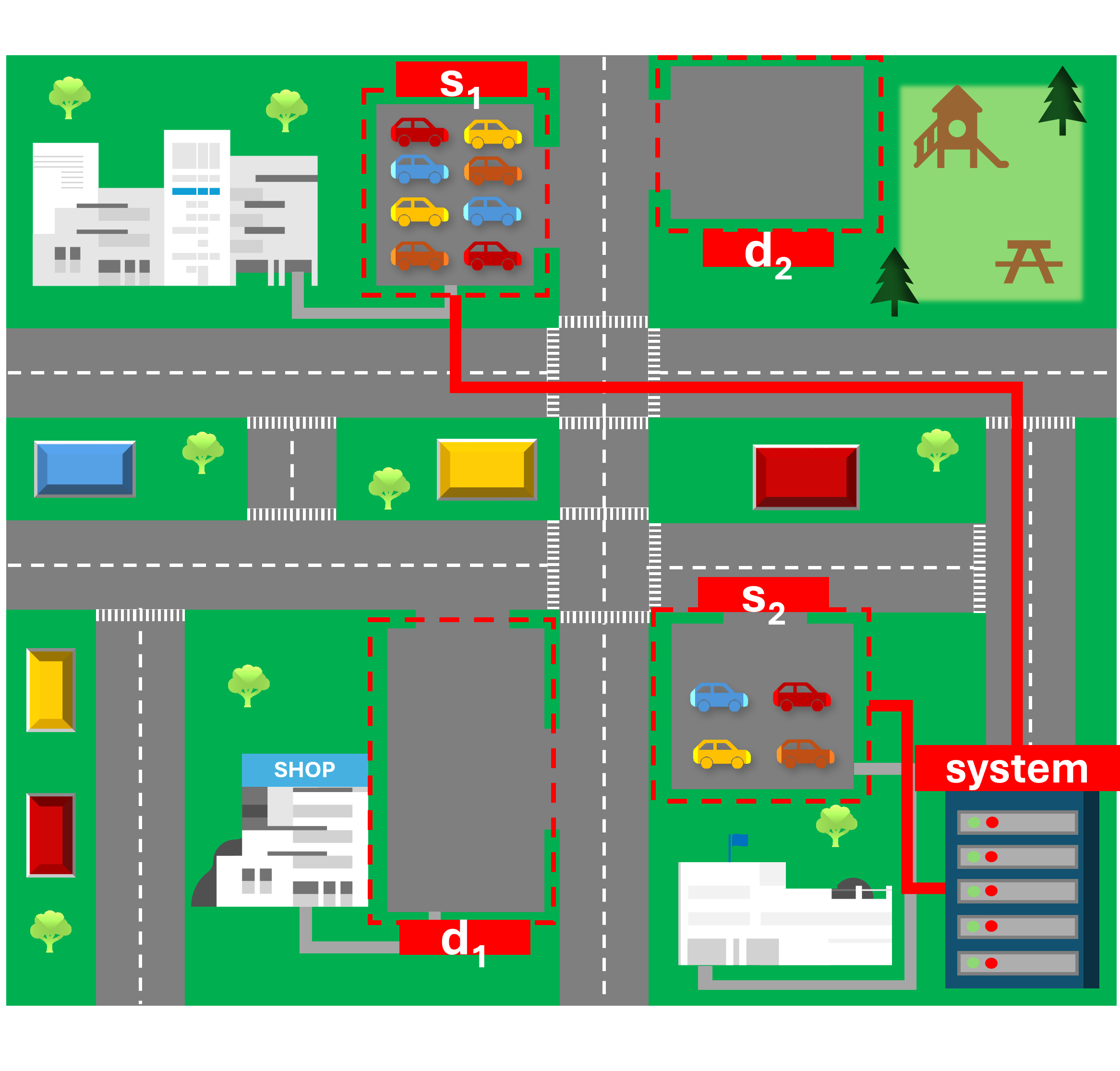}
\vspace{-0.15in}
\caption{Smart Transportation System Model.}
\label{Img: System Model}
\vspace{-0.15in}
\end{figure}

\subsection{System Model \label{Sect: System Model}}
Let the transportation network be represented as a graph $\mathcal{G} = \{\mathcal{V}, \mathcal{E}\}$, where $\mathcal{V}$ represents a set of locations or intersections, and the set of directed edges $\mathcal{E}$ represents the directed roadways between physical locations in $\mathcal{V}$. There exists a set of $k$ source-destination pairs $(s_1, d_1), \cdots, (s_k, d_k)$, called \textit{commodities}. For each commodity $i = 1, \cdots, k$, there is an associated finite demand $r_k$, where $\sum_{i = 1}^k r_k = r$ denotes the entire network demand. Each commodity demand $r_i$ corresponds to the amount of demand with source $s_i$ and destination $d_i$. For each commodity, the nonempty set of simple paths from $s_i$ to $d_i$ is $\mathcal{P}_i$, and $\mathcal{P} = \cup_i \mathcal{P}_i$ is the set of all network paths. Such commodities are illustrated in Figure \ref{Img: System Model} with a red dashed box and labels $s_i$ and $d_i$ at their corresponding origin and destination, with various paths available between each origin and destination.

Let $\boldsymbol{f} = \{\boldsymbol{f}_P\}_{\forall P \in \mathcal{P}}$ be the network flow vector, also referred to as an assignment of demands to network paths. Similarly, the flow vector of a commodity $(s_i, d_i)$ is denoted by $\boldsymbol{f}_i = \{f_P\}_{\forall P \in \mathcal{P}_i}$. For an edge $e$, the edge flow volume $f_e = \sum_{P: P\in e} f_P$. For a path $P$, the path flow volume $f_P$ is defined as $\sum_{e \in P} f_e$. Furthermore, a network flow $\boldsymbol{f}$ is \textit{feasible} if and only if for all commodities $\sum_{P \in \mathcal{P}_i} f_P = r_i$.

Using the edge flow $f_e$, the travel time (latency) on edge $e$ as a result of the congestion caused by the edge flow $f_e$ can be estimated using the Bureau of Public Roads (BPR) cost function as shown below:

\vspace{-0.15in}
\begin{equation}
\tau(f_e) = t^{ff}_e \left (1 + \lambda \left( \frac{f_e}{c_e}\right)^\beta \right),
\label{BPR}
\end{equation}

\vspace{-0.1in}
\noindent where $t^{ff}_e$ denotes the free-flow travel time along an edge $e$,  $c_e$ denotes the flow-capacity of an edge, and $\lambda$ and $\beta$ are shape coefficients which are commonly assumed to be 0.15 and 4 respectively. The latency of a path $P$ with  flow volume $f_p$ is defined as $\tau(f_P) = \sum_{e \in P} \tau(f_e)$. Assume that both $t^{tt}_e$ and $c_e$ are known to the system and travelers comprising the network demand at all edges $e \in \mathcal{E}$. Let $(G, r)$ denote an instance of the Stackelberg routing problem on network $G$ with total demand $r$, where the latency function is given by Equation (\ref{BPR}).

Suppose that the total amount of network travel demand $r$ is comprised of $N$ groups of homogeneous travelers, denoted by $\boldsymbol{r} = \{ r_{\alpha_1}, \cdots, r_{\alpha_N} \}$, where $\sum_{i = 1}^N r_{\alpha_i} = r$. Here, $\alpha \in [0,1]$ is a parameter representing a demand group's trust in the system (i.e., the probability that a demand group will choose to take a path recommended by the system), where trust increases as $\alpha \rightarrow 1$. Suppose each homogeneous travel demand group consists of travelers with similar trust in the system, $\alpha_i$, and similar prior beliefs about the flow on each path, $q_i$. For the sake of simplicity, assume $\alpha_1 < \alpha_2 < \cdots < \alpha_N$.

\subsection{Problem Formulation \label{Problem Fomulation}}
The objective of the compliant demand group is to split its flow in accordance with the corresponding path recommendations shared by the system, forming an obedient flow $\boldsymbol{f}^C$. On the other hand, the noncompliant fraction of demand aims to split its flow along paths of minimum latency with respect to its prior belief, forming a selfish flow, $\boldsymbol{f}^{NC}$, where each demand chooses a path $P_i^{NC}$ based on its prior belief $q_i$.

Suppose all network demand communicates with an SRP (a.k.a., the system), as shown in Figure \ref{Img: System Model}. Let $P_{i,s}$ be a path recommendation sent from the system to a demand group, and let $\boldsymbol{P}_s = \{P_{1,s}, \cdots, P_{N,s}\}$ be the set of all path recommendations sent to all corresponding demand groups. After receiving a recommendation from the system, each demand chooses a path (or paths, if the demand is split). For partially-compliant demands, the $i^{th}$ demand group accepts the recommendation $P_{i,s}$ with probability $\alpha_i$ and rejects the recommendation with probability $(1 - \alpha_i)$, choosing $P_i^{NC}$ instead. Let $\boldsymbol{f}^*$ denote the final flow of the total network demand.

The system's objective is to share a set of path recommendations $\boldsymbol{P}_s$ that induce flow $\boldsymbol{f}^*(\boldsymbol{P}_s)$ of the total network demand, which results in congestion (i.e., minimum total network travel time) which can be computed as 

\vspace{-0.1in}
\begin{equation}
    C(\boldsymbol{f}^*(\boldsymbol{P}_s)) = \sum_{P \in \mathcal{P}} \boldsymbol{f}_P^*(\boldsymbol{P}_s) \cdot \tau(\boldsymbol{f}_P^*(\boldsymbol{P}_s)).
\label{Eqn: Congestion}
\end{equation}
\vspace{-0.1in}

Note that $\boldsymbol{f}^*$ is dependent on $\boldsymbol{P}_s$ along with the prior beliefs $q_i$ and trust values $\alpha_i$ for each $i = 1, \cdots, N$ demand group. Therefore, the system will construct the best-response path-profile recommendation $\boldsymbol{P}_s^*$, as shown below:
\begin{equation}
\begin{array}{lcll}
\boldsymbol{P}_s^*  & = & \displaystyle \argmin_{\boldsymbol{P}_s \in \mathcal{P}} & C( \boldsymbol{f}^*(\boldsymbol{P}_s))
\\[3ex]
&& \text{subject to} & \text{1. } \displaystyle f_P \geq 0
\\[1.5ex]
&&& \text{2. } \displaystyle \sum_{P \in \mathcal{P}} f_P = r
\\[3ex]
&&& \text{3. } \displaystyle \sum_{P \in \mathcal{P}} f_P \cdot \delta_P =  f_e, \ \forall \text{ } e \in \mathcal{E},
\end{array}
\tag{P1}
\label{Prob: Optimal Path Recs Sender}
\end{equation}
where $\delta_{P}$ is the link-path incidence taking a value of 1 when edge $e$ belongs to path $P$ and a value of 0 otherwise.

The \textit{Stackelberg equilibrium} of the game between the system and total network demand is defined as the pair $(\boldsymbol{P}_s^*, \boldsymbol{f}^*(\boldsymbol{P}_s^*))$ where $\boldsymbol{P}_s^*$ is computed as shown in Equation \ref{Prob: Optimal Path Recs Sender} and $\boldsymbol{f}^*(\boldsymbol{P}_s^*)$ denotes the best response paths of demand group when the system shares $\boldsymbol{P}_s^*$.

\section{Proposed Algorithm \label{Sec: Proposed Methodology}}
Given that the Stackelberg game described in Section \ref{Sect: System Model} is equivalent to a linear bilevel programming problem \cite{colson2007overview}, it is NP-Hard to solve the Problem \ref{Prob: Optimal Path Recs Sender}. Moreover, even for networks with parallel links and linear latency functions, computing the opimal Stackelberg routing strategy is NP-Hard \cite{roughgarden2001stackelberg}. Given that Problem \ref{Prob: Optimal Path Recs Sender} assumes polynomial latency functions, it is preferable to compute the approximate best response recommendation path-profile of the sender as opposed to computing the exact optimal path-profile recommendations. 

\subsection{Approximate Best-Response Path-Profile Computation}
The proposed Stackelberg routing is a greedy strategy that prioritizes demand groups with greatest trust, saturating paths with smallest travel time under a special instance of the graph $(G,r^{\circledast})$, where $r^{\circledast} = r_{\alpha = 1}$ a travel demand group consisting of a completely compliant demand. The proposed algorithm is called TASR. As in most opt-restricted strategies, the optimal path assignment for the case in which the entire demand is compliant is computed. The problem of finding an optimal flow $\boldsymbol{f}^{\circledast}$ for an instance  $(G,r^{\circledast})$ is formulated as a convex optimization problem \cite{beckmann1956studies} as follows:
\begin{equation}
\begin{array}{ll}
\text{minimize} & CC(\boldsymbol{f}^{\circledast}) \triangleq \displaystyle \sum_{e \in \mathcal{E}} f^{\circledast}_e \cdot \tau(f^{\circledast}_e)
\\[3ex]
\text{subject to} & \text{1. } \displaystyle f^{\circledast}_P \geq 0
\\[1.5ex]
& \text{2. } \displaystyle \sum_{P \in \mathcal{P}} f^{\circledast}_P = r^{\circledast}
\\[3ex]
& \text{3. } \displaystyle \sum_{P \in \mathcal{P}} f^{\circledast}_P \cdot \delta_P =  f^{\circledast}_e \text{  } \forall \text{ } e \in \mathcal{E},
\end{array}
\label{Eqn: SO Flows Subgraph}
\end{equation}
where $\delta_{P}$ is the link-path incidence taking a value of 1 when edge $e$ belongs to path $P$ and a value of 0 otherwise.


The proposed method for finding an approximately optimal set of path recommendations for a single-commodity instance $(G, \boldsymbol{r})$ is stated in Algorithm \ref{alg:tasr}. Note that while the first step in TASR algorithm is to compute $\boldsymbol{f}^{\circledast}$, the paths in $\boldsymbol{f}^{\circledast}$ are ordered by \textit{increasing} latency. The system then assumes the noncompliant demand will choose paths in a selfish manner. In response to the probable flow of the noncompliant demand, the system greedily assigns path recommendations in order of increasing latency to the partially-compliant demand groups. Note that by assigning path recommendations in the order of smallest latency first to the partially compliant demands, where demands are ordered by increasing trust, TASR attempts to maximize the probability that a recommended path is chosen by its corresponding partiallly-compliant demand group. Additionally, this guarantees that paths of largest latency are not unnecessarily recommended to demands, thereby increasing network congestion may they be accepted. Finally, the compliant demand is assigned to the remaining unsaturated paths in order of smallest latency until no compliant demand remains.

For the multi-commodity setting, the Multi-Commodity TASR algorithm (Algorithm \ref{alg:mc tasr}) is an extension of the Single-Commodity TASR that prioritizes commodities according to largest fraction of noncompliant demand. Since the system has greater control over partially-compliant and fully compliant demands, this allows the system to iteratively mitigate the selfish path choices of the noncompliant demand for each commodity by assigning compliant demands to paths unfavorable to the noncompliant demand. For each commodity, TASR is run, and the resulting demand on any edges that are shared (i.e., overlap) with those of other commodities is assumed to be existing demand in network of those subsequent commodities.

\subsection{TASR Properties}


In this section, properties of the performance of the proposed TASR algorithm are presented.

\begin{definition}
    For an instance $(G, r)$, the \textit{efficiency ratio} of TASR($G, r)$ is the ratio of the cost of the resulting path-profile recommendation $\boldsymbol{P}_s$ and the cost of the solution obtained by CC($G, r^{\circledast}$). In other words, the efficiency ratio is given by $\frac{C(\boldsymbol{f}^*(\boldsymbol{P}_s))}{CC(\boldsymbol{f}^{\circledast})}$.
\end{definition}

\begin{theorem}
The optimal path-profile recommendation of the system is a path-profile recommendation $\boldsymbol{P}_s^*$ such that
\begin{equation}
\frac{C(\boldsymbol{f}^*(\boldsymbol{P}_s))}{CC(\boldsymbol{f}^{\circledast})} = 1.
\end{equation}
\end{theorem}
\begin{proof}
Given Equation \ref{Eqn: SO Flows Subgraph}, $\boldsymbol{f^{\circledast}}$ is an optimal solution to the special case of Problem \ref{Prob: Optimal Path Recs Sender} for an instance $(G, r^{\circledast})$ in which all demand is completely compliant. If TASR($G, r$) results in a path-profile recommendation $\boldsymbol{P}_s$ such that the efficiency ratio is 1, the path-profile recommendation $\boldsymbol{P}_s$ must be optimal. 
\end{proof}

\begin{definition}
    A flow vector $\boldsymbol{f}^{\prime}$ is a \textit{subflow} of the network flow $\boldsymbol{f}$ if and only if $f^{\prime}_P \leq f_P, \forall P \in \mathcal{P}$.
\label{Def: Subflow}
\end{definition}

\begin{corollary}
\label{Cor: NC Subflow}
If $\boldsymbol{f}^{NC}$ is not a subflow of $\boldsymbol{f}^{\circledast}$, the efficiency ratio of TASR($G, r$) will be at least 
\begin{equation}
    \frac{CC(\boldsymbol{f}^{\circledast}) + C(\boldsymbol{f}^{NC})}{CC(\boldsymbol{f}^{\circledast})}.
\end{equation}
\end{corollary}

\begin{proof}
    From Definition \ref{Def: Subflow}, if $\boldsymbol{f}^{NC}$ is not a subflow of $\boldsymbol{f}^{\circledast}$, then $f^{NC}_P > f^{\circledast}_P$, for at least one path $P \in \mathcal{P}$. Therefore, the cost of the solution to TASR($G, r$) will always be at least $CC(\boldsymbol{f}^{\circledast}) + C(\boldsymbol{f}^{NC})$, since the paths of the noncompliant demand groups cannot be influenced by the system.
\end{proof}


\begin{algorithm}
\caption{Single-Commodity TASR$(G, \boldsymbol{r})$ }\label{alg:tasr}
\begin{algorithmic}[1]

\vspace{1ex}
\Statex \textit{// Compute optimal flow for complete compliance instance}
\State $\boldsymbol{f}^{\circledast}$ = \Call{CC}{$G, r^{\circledast}$}
\vspace{1ex}

\Statex \textit{// Sort optimal $\boldsymbol{f}^{\circledast}$ paths by increasing latency}
\State $\boldsymbol{f}^{\circledast}$ = \Call{SortByLatency}{$\boldsymbol{f}^{\circledast}$, reverse = True} 
\vspace{1ex}

\Statex \textit{// Predict flow chosen by noncompliant demand}
\State $\boldsymbol{f}^{NC}$ = \Call{GetChosenPath}{$r^{\alpha_1}$}
\vspace{1ex}

\Statex \textit{// Consider next unsaturated path}
\State $P_{i,s}$ = \Call{Sample}{$f^{\circledast}$}
\State $\boldsymbol{P}^*[r^{\alpha_1}] = P_{i,s}$
\\
\Statex \textit{// Assume in resulting flow, demand acts selfishly}
\For{$i$ \textbf{in} $\boldsymbol{f}^{NC}$}
\State $\boldsymbol{\hat{f}}^*[i]$ = $\boldsymbol{f}^{NC}[i]$
\EndFor
\\
\Statex \textit{// For each partially-compliant demand}
\For{$i$ \textbf{in} $r_{\alpha_2}, \cdots, r_{\alpha_{N-1}}$}
    \State $P_{i,s}$ = \Call{Sample}{$f^{\circledast}$}
    \State $P_i^{NC}$ = \Call{CheckRecommendation}{$q_i$, $P_{i,s}$}
    \Statex \textit{\hspace{18pt}// If demand likely to accept recommendation}
    \If{\Call{CheckDecision}{$i$, $P_{i,s}$} = $P_{i,s}$}
    \Statex \textit{\hspace{30pt} // assume demand chooses path}
    \State $\hat{f}^*$[$P_{i,s}$] += $i$
    \Statex \textit{\hspace{30pt} // recommend path to demand}
    \State $\boldsymbol{P}^*[i] = P_{i,s}$
\Else
\If{\Call{CheckDecision}{$i$, $P_{i,s}$} = $P_i^{NC}$}
\Statex \textit{\hspace{30pt} // assume demand acts selfishly}
\State $\hat{f}^*$[$P_i^{NC}$] += $i$ 
\Statex \textit{\hspace{30pt} // recommend path to demand}
    \State $\boldsymbol{P}^*[i] = P_{i,s}$
\EndIf
\EndIf
\EndFor
\\
\Statex \textit{// Saturate unsaturated paths of $f^{\circledast}$ with compliant demand}
\For{$i$ \textbf{in} $\boldsymbol{f}^{\circledast}$}
\If{\textbf{not }\Call{IsSaturated}{$f^*[i]$}}
\State $\hat{f}^*[i]$ = $r^{\alpha_N}$    \State $\boldsymbol{P}^*[ r^{\alpha_N}] = i$
\EndIf
\EndFor
return $\boldsymbol{P}^*$
\end{algorithmic}
\label{Alg: TASTT}
\end{algorithm}

\begin{algorithm}
\caption{Multi-Commodity TASR$(\mathcal{G}, \mathcal{R})$}\label{alg:mc tasr}
\begin{algorithmic}[1]
\Statex \textit{// Sort commodities in descending order of fraction of noncompliant demand}
\State Instances = \Call{SortByNonCompliant}{$\mathcal{G}, \mathcal{R}$}
\\
\For{$i$ in Instances}
\State NextInstance = Instances[i:]
\State $\boldsymbol{P}_i^*$ = \Call{TASR}{$G, \boldsymbol{r}$}
\For{$j$ in NextInstances}
\If{\Call{ShareEdges}{$i, j$}}
\State \Call{UpdateDemands}{$\boldsymbol{P}_i^*$, j}
\EndIf
\EndFor
\EndFor
\State return $\boldsymbol{P}^*$
\end{algorithmic}
\end{algorithm}

\section{Trust Dynamics \label{Sect: Trust Dynamics}}
This section focuses on how the $i^{th}$ traffic demand group updates its trust $\alpha_i$ in the system after traversing the network. Given that human trust is dynamic and changes through repeated interactions with a system, it is assumed that the system will have previously interacted with each traffic demand, observed their interaction-to-interaction path choices, and approximated each demand's trust $\alpha_i$ and prior belief $q_i$.

Let $P_i^*$ be a path that minimizes the expected travel time of $i^{th}$ traffic demand group with respect to its prior belief $q_i$.
\begin{equation}
P_i^* = \argmin_{P \in \mathcal{P}} q_i(f_{P_i}) \cdot \tau(f_{P_i}).
\end{equation}

Given $\alpha_i$, the $i^{th}$ traffic demand group will either accept the recommendation from the system $P_{i,s}$ or reject the recommendation, choosing $P_i^*$ instead. Let $P_i^{\circledast}$ denote the path chosen by the $i^{th}$ demand group after receiving a recommendation from the system. After traversing the network, the $i^{th}$ demand group will experience some regret $B_i$ for having interacted with the system rather than relying only on its prior belief to choose a path through the network. Let $B_i$ be defined as:

\begin{equation}
B_i = \tau(f_{P_i^{\circledast}}) - \tau(f_{P_i^*}) 
\label{Eqn: Regret}
\end{equation} 

The trust dynamics of $i^{th}$ demand group are modeled As:
\begin{equation}
\alpha^+ = 
\begin{cases}
  \text{min}\{1, \alpha_i - \varepsilon \cdot \sign(\nabla B_i)\} & \text{if }B_i \leq 0\\    
  \text{max}\{\alpha_i - \varepsilon \cdot \sign(\nabla B_i), 0\}  & \text{otherwise},  
\end{cases}
\label{Eqn: Trust Update}
\end{equation}
where $\varepsilon$ is a fixed step size. $\varepsilon$ can be interpreted as a penalty that a demand applies to its trust when it receives unfavorable recommendations, or a reward when it receives favorable recommendations.

\section{Evaluation Methodology and Datasets \label{Sect: Evaluation Methodology}}

\subsection{Performance Metrics}


The performance of the proposed TASR strategy is evaluated based on network congestion, where congestion is the total travel time experienced by all demand in the network (refer to Equation \ref{Eqn: Congestion}). The congestion of the network in the scenario in which all network demand is centrally routed is used as a baseline. The network congestion obtained under TASR is compared to that of four prominent Stackelberg routing algorithms, as noted in Section \ref{Sect: Experiments, Datasets, Preprocessing} under various classes of demand (in terms of trust), amounts of total demand, and different real-world networks. 

Trust is also considered as a metric of performance, where the trust of a demand group after interacting with the system is calculated according to Equation \ref{Eqn: Trust Update}.

\subsection{Experiments, Datasets, and Preprocessing \label{Sect: Experiments, Datasets, Preprocessing}}

Four transportation networks are simulated in order to evaluate the performance of the TASR algorithm against the following Stackelberg routing strategies: LLF \cite{roughgarden2001stackelberg}, Scale \cite{roughgarden2001stackelberg}, Aloof \cite{roughgarden2001stackelberg}, and ASCALE \cite{bonifaci2007impact}. The proposed method is implemented in Python 3.10.9 using primarily the NumPy library. An implementation of the Frank-Wolfe algorithm \cite{matteo2021} is used to compute optimal traffic flows for networks with only compliant travelers, abbreviated as CC in each figure. 

Networks simulated include Sioux Falls, Chicago Sketch, and Sydney. Specific information regarding the each network is maintained by the Transportation Networks for Research Core Team \cite{stabler2020} and publicly available. In the single-commodity setting, the Sioux Falls and Chicago Sketch networks are considered. For the Sioux Falls network, the physical link parameters of \cite{kaddoura2014optimal} are followed. In the multi-commodity setting, reasonable traffic demands (i.e., demand not exceeding the capacity of the network) were simulated on the Sioux Falls, Chicago sketch, Austin, and Sydney, and the network congestion of each was compared for the TASR, LLF, Scale, ASCALE, and Aloof algorithms. For all experiments, five different classes of traffic demand are considered, with the following corresponding trust parameters: $\alpha_1 = 0$, $\alpha_2 = 0.25$, $\alpha_3 = 0.5$, $\alpha_4 = 0.75$, and $\alpha_5 = 1$. For each network, a value $\Delta$ is considered, where $\Delta$ denotes the average amount of demand per network edge.

The following modifications to the LLF, Scale, ASCALE, and Aloof algorithms were made in order to implement them for comparison to TASR. Given that LLF, Scale, ASCALE, and Aloof do not consider partially-compliant agents, to implement each algorithm, demand groups with $\alpha_i \geq 0.5$ are considered compliant, while all other demand groups are considered non-compliant. Additionally, in the implementation of ASCALE, an ideal scaling value of $\rho = 1 + \sqrt{1 - \hat{\mu}}$, following from \cite{swamy2012}, was considered for preliminary experimentation, where $\hat{\mu}$ denotes the total fraction of compliant demand. In each experiment, each demand was scaled by $\rho$ and rounded to the nearest integer value. For the sake of fairness in Scale and ASCALE, since the compliant demand groups have differing trusts in the system, each scaled optimal path demand in each experiment was recommended to an equal fraction of each compliant demand group. In the implementation of Aloof, each experimental total network demand was scaled by $\hat{\mu}$.


\subsubsection{Single-Commodity Networks\label{Sect: Single Commodity Networks}}
Given the significance of parallel-path networks in the Stackelberg routing literature, two single-commodity networks were considered for validating the performance of TASR algorithm, and subnetworks of the Sioux Falls network and the Chicago Sketch network. 

The first single-commodity network considered is a subgraph of the Sioux Falls network with a single origin-destination pair $(20, 10)$ with four available parallel paths comprised of 16 edges. On this network, three scenarios are simulated: $\Delta=5$, $\Delta=10$, and $\Delta=15$. In each scenario, one-sixth of the network demand is fully compliant and noncompliant respectively. 


\begin{table*}[!t]
\caption{Comparison of Average Congestion and Runtime (in Seconds) on Different Single-Commodity Road Networks over 1000 Iterations.}
\label{tab: SC Travel Time}
\centering
\sisetup{detect-weight=true,detect-inline-weight=math}
\begin{tabular}{c c c c c c c c}
\hline
\\[-1.5ex]
\multicolumn{1}{c}{} & \multirow{2}*{\textbf{Attribute}} & \textbf{TASR} & \textbf{CC*} & \textbf{LLF} & \textbf{ASCALE} & \textbf{Scale} & \textbf{Aloof}
\\
\multicolumn{1}{c}{} &  &  Mean (SD) & Mean (SD) & Mean (SD) & Mean (SD) & Mean (SD) & Mean (SD)
\\[1ex]
\hline\hline
\\[-1.5ex]
\multirow{2}[8]*{\textbf{Sioux Falls}} & $\Delta = 5$ & \bfseries \num[round-precision=3,round-mode=figures,
scientific-notation=true]{1542.80} (\num[round-precision=2,round-mode=figures,
scientific-notation=true]{24.24}) & \num[round-precision=3,round-mode=figures,
scientific-notation=true]{1520.02} (\num[round-precision=2,round-mode=figures,
scientific-notation=true]{0.0})  & \num[round-precision=3,round-mode=figures,
scientific-notation=true]{1546.54} (\num[round-precision=2,round-mode=figures,
scientific-notation=true]{27.92}) & \num[round-precision=3,round-mode=figures,
scientific-notation=true]{1546.77} (\num[round-precision=2,round-mode=figures,
scientific-notation=true]{28.09}) & \num[round-precision=3,round-mode=figures,
scientific-notation=true]{1545.32} (\num[round-precision=2,round-mode=figures,
scientific-notation=true]{27.14}) & \num[round-precision=3,round-mode=figures,
scientific-notation=true]{1783.76} (\num[round-precision=2,round-mode=figures,
scientific-notation=true]{189.07})
\\[1ex]
\multirow{2}[8]*{(Small)} & Runtime & \num[round-precision=2,round-mode=figures,
scientific-notation=true]{0.00044552016258239745} & \num[round-precision=2,round-mode=figures,
scientific-notation=true]{0.01462411880493164} & \bfseries \num[round-precision=2,round-mode=figures,
scientific-notation=true]{0.0003933103084564209} & \num[round-precision=2,round-mode=figures,
scientific-notation=true]{0.0006014521121978759}  & \num[round-precision=2,round-mode=figures,
scientific-notation=true]{0.0005173232555389405} & \num[round-precision=2,round-mode=figures,
scientific-notation=true]{0.000557107925415039}
\\[-1ex]
& \cline{1-7}
\\[-1ex]
& $\Delta = 10$ & \bfseries \num[round-precision=3,round-mode=figures,
scientific-notation=true]{3087.50} (\num[round-precision=2,round-mode=figures,
scientific-notation=true]{52.57}) & \num[round-precision=3,round-mode=figures,
scientific-notation=true]{3040.48} (\num[round-precision=2,round-mode=figures,
scientific-notation=true]{0.0}) & \num[round-precision=3,round-mode=figures,
scientific-notation=true]{3092.66} (\num[round-precision=2,round-mode=figures,
scientific-notation=true]{54.44}) & \num[round-precision=3,round-mode=figures,
scientific-notation=true]{3093.42} (\num[round-precision=2,round-mode=figures,
scientific-notation=true]{54.83})  & \num[round-precision=3,round-mode=figures,
scientific-notation=true]{3090.06} (\num[round-precision=2,round-mode=figures,
scientific-notation=true]{55.79}) & \num[round-precision=3,round-mode=figures,
scientific-notation=true]{3601.13} (\num[round-precision=2,round-mode=figures,
scientific-notation=true]{353.60})
\\[1ex]
 & Runtime & \bfseries \num[round-precision=2,round-mode=figures,
scientific-notation=true]{0.0003835761547088623} & \num[round-precision=2,round-mode=figures,
scientific-notation=true]{0.05498766899108887} & \num[round-precision=2,round-mode=figures,
scientific-notation=true]{0.0004148082733154297} & \num[round-precision=2,round-mode=figures,
scientific-notation=true]{0.0006085138320922851} & \num[round-precision=2,round-mode=figures,
scientific-notation=true]{0.0005823702812194824}  & \num[round-precision=2,round-mode=figures,
scientific-notation=true]{0.0006254587173461914}
\\[1ex]

\hline
\\[-1ex]
\multirow{2}[8]*{\textbf{Chicago Sketch}} & $\Delta = 5$ & \bfseries \num[round-precision=3,round-mode=figures,
scientific-notation=true]{6619.97} (\num[round-precision=2,round-mode=figures,
scientific-notation=true]{415.10}) & \num[round-precision=3,round-mode=figures,
scientific-notation=true]{6190.84} (\num[round-precision=2,round-mode=figures,
scientific-notation=true]{0.0}) & \num[round-precision=3,round-mode=figures,
scientific-notation=true]{6628.967} (\num[round-precision=2,round-mode=figures,
scientific-notation=true]{0.00042742586135864256}) & \num[round-precision=3,round-mode=figures,
scientific-notation=true]{6668.102450754354} (\num[round-precision=2,round-mode=figures,
scientific-notation=true]{440.08}) & \num[round-precision=3,round-mode=figures,
scientific-notation=true]{6624.05} (\num[round-precision=2,round-mode=figures,
scientific-notation=true]{390.40}) & \num[round-precision=3,round-mode=figures,
scientific-notation=true]{7446.01}(\num[round-precision=2,round-mode=figures,
scientific-notation=true]{876.20})
\\[1ex]
\multirow{2}[8]*{(Medium)} & Runtime & \bfseries \num[round-precision=2,round-mode=figures,
scientific-notation=true]{0.0004139842987060547} & \num[round-precision=2,round-mode=figures,
scientific-notation=true]{0.12822628021240234} & \num[round-precision=2,round-mode=figures,
scientific-notation=true]{0.00042742586135864256} & \num[round-precision=2,round-mode=figures,
scientific-notation=true]{0.0007459816932678222} & \num[round-precision=2,round-mode=figures,
scientific-notation=true]{0.0007505965232849121} & \num[round-precision=2,round-mode=figures,
scientific-notation=true]{0.0007127838134765625}
\\[-1ex]
& \cline{1-7}
\\[-1ex]
& $\Delta = 10$ & \bfseries \num[round-precision=3,round-mode=figures,
scientific-notation=true]{13182.31} (\num[round-precision=2,round-mode=figures,
scientific-notation=true]{822.30}) & \num[round-precision=3,round-mode=figures,
scientific-notation=true]{12382.83} (\num[round-precision=2,round-mode=figures,
scientific-notation=true]{0.0}) & \num[round-precision=3,round-mode=figures,
scientific-notation=true]{13300.47} (\num[round-precision=2,round-mode=figures,
scientific-notation=true]{797.50}) & \num[round-precision=3,round-mode=figures,
scientific-notation=true]{13282.22} (\num[round-precision=2,round-mode=figures,
scientific-notation=true]{855.41}) & \num[round-precision=3,round-mode=figures,
scientific-notation=true]{13290.45} (\num[round-precision=2,round-mode=figures,
scientific-notation=true]{831.67}) & \num[round-precision=3,round-mode=figures,
scientific-notation=true]{15006.66} (\num[round-precision=2,round-mode=figures,
scientific-notation=true]{1663.47})
\\[1ex]
& Runtime & \num[round-precision=2,round-mode=figures,
scientific-notation=true]{0.0005026636123657227} & \num[round-precision=2,round-mode=figures,
scientific-notation=true]{0.09341907501220703} & \bfseries \num[round-precision=2,round-mode=figures,
scientific-notation=true]{0.0003577713966369629} & \num[round-precision=2,round-mode=figures,
scientific-notation=true]{0.0007834534645080566} & \num[round-precision=2,round-mode=figures,
scientific-notation=true]{0.0007111339569091797} & \num[round-precision=2,round-mode=figures,
scientific-notation=true]{0.0008128609657287598}
\\[1ex]
\hline
\end{tabular}
{\\[1ex] \footnotesize $^*$CC denotes a special case of TASR in which all demand is completely compliant.}
\end{table*}

\begin{table}[!t]
\caption{Average travel time (in mins) per unit demand in single-commodity Sioux Falls network.}
\label{tab: SF Single Commodity Travel Times}
\centering
\begin{tabular}{c c c c}
\hline
\\[-1.5ex]
\multicolumn{1}{c}{}{} & \textbf{$\Delta = 5$} & \textbf{$\Delta = 10$} & \textbf{$\Delta = 15$}
\\[0.5ex]
\hline \hline
\\[-1ex]
\textbf{TASR} & \textbf{19.285} & \textbf{19.297} & \textbf{19.296} 
\\[1ex]
\textbf{CC*} & 19.000 & 19.003 & 19.014
\\[1ex]
\textbf{LLF} & 19.332 & 19.329 & 19.351
\\[1ex]
\textbf{ASCALE} & 19.335 & 19.334 & 19.339
\\[1ex]
\textbf{Scale} & 19.317 & 19.332 & 19.324
\\[1ex]
\textbf{Aloof} & 22.297 & 22.507 & 22.526
\\[1ex]
\hline
\end{tabular}
{\\[1ex] \footnotesize $^*$CC denotes a special case of TASR in which all demand is completely compliant.}
\vspace{-3ex}
\end{table}

\begin{figure}[!t]
 \vspace{-0.1in}
\centering
\captionsetup{justification=centering}
\includegraphics[width=0.475\textwidth]{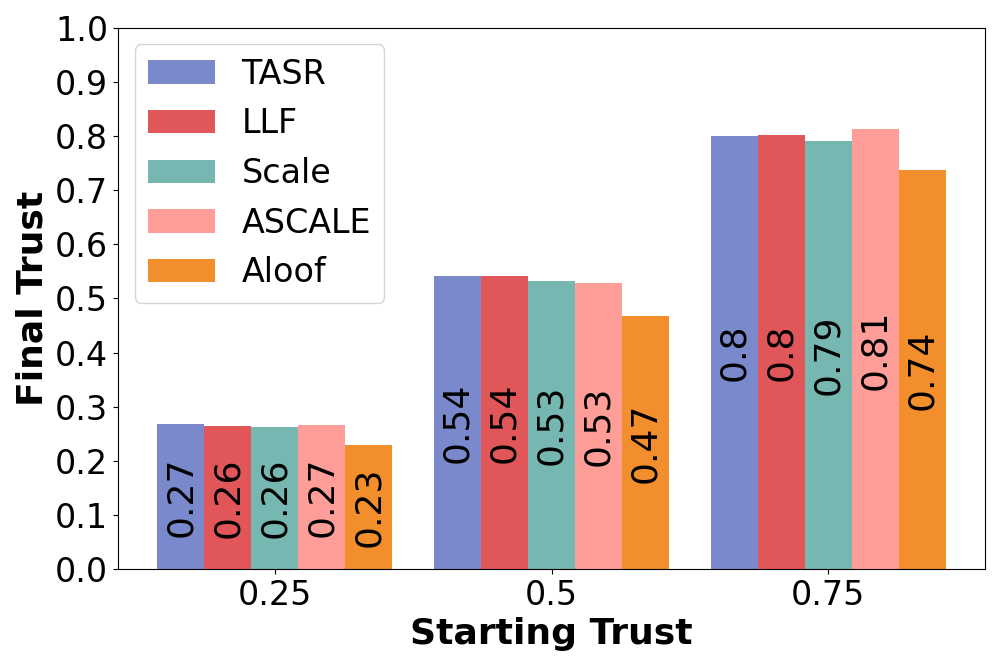}
\caption{Comparison of Starting Trust to Final Trust for each Algorithm on the Single-Commodity Sioux Falls Network, Averaged over 500 Independent Iterations with $\varepsilon = 0.25$.}
\label{Img: SF Subnet Avg Trust}
\end{figure}

\begin{figure}[!t]
\centering
\captionsetup{justification=centering}
\includegraphics[width=0.48\textwidth]{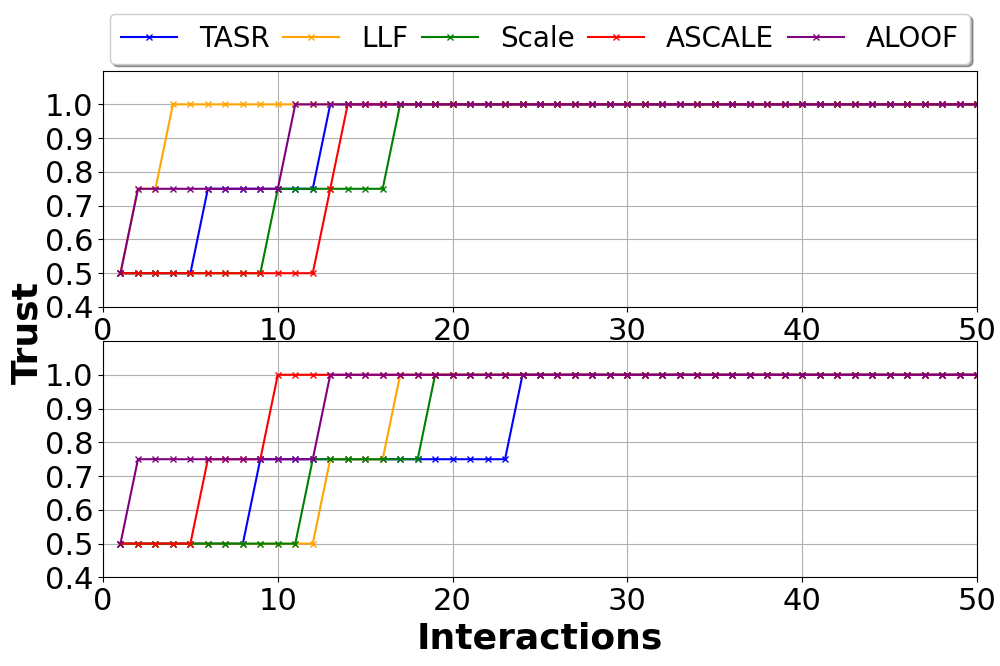}
\caption{Comparison of Trust Dynamics over 50 Repeated Interactions on the Single-Commodity Sioux Falls Network with $\Delta = 10$ (Top) and $\Delta = 20$ (Bottom), when Trust is Initialized at 0.5 for All Agents with $\varepsilon = 0.25$.}
\label{Img: SF Subnet Repeated Trust 0.5 Delta 10 and 20}
\end{figure}

\begin{table*}[!t]
\caption{Comparison of Average Congestion and Standard Deviation on Different Multi-Commodity Road Networks.}
\label{tab: MC Travel Time}
\centering
\sisetup{detect-weight=true,detect-inline-weight=math}
\begin{tabular}{c c c c c c c c}
\hline
\\[-1.5ex]
\multicolumn{1}{c}{} & \textbf{Demand} & \textbf{TASR} & \textbf{CC*} & \textbf{LLF} & \textbf{ASCALE} & \textbf{Scale} & \textbf{Aloof}
\\
\multicolumn{1}{c}{} & \textbf{Rate} &  Mean (SD) & Mean (SD) & Mean (SD) & Mean (SD) & Mean (SD) & Mean (SD)
\\[1ex]
\hline\hline
\\[-1.5ex]
\multirow{2}*{\textbf{Sioux Falls}} & $\Delta = 5$ & \bfseries \num[round-precision=3,round-mode=figures,
scientific-notation=true]{9931.30} (\num[round-precision=2,round-mode=figures,
scientific-notation=true]{548.36})  & \num[round-precision=3,round-mode=figures,
scientific-notation=true]{8211.36} (\num[round-precision=2,round-mode=figures,
scientific-notation=true]{432.92}) & \num[round-precision=3,round-mode=figures,
scientific-notation=true]{9944.20} (\num[round-precision=2,round-mode=figures,
scientific-notation=true]{555.76}) &  \num[round-precision=3,round-mode=figures,
scientific-notation=true]{10051.85} (\num[round-precision=2,round-mode=figures,
scientific-notation=true]{595.09}) & \num[round-precision=3,round-mode=figures,
scientific-notation=true]{9939.61} (\num[round-precision=2,round-mode=figures,
scientific-notation=true]{584.45}) & \num[round-precision=3,round-mode=figures,
scientific-notation=true]{15394.04} (\num[round-precision=2,round-mode=figures,
scientific-notation=true]{721.60})
\\[1ex]
(Small) & $\Delta = 10$ & \bfseries \num[round-precision=3,round-mode=figures,
scientific-notation=true]{17491.09} (\num[round-precision=2,round-mode=figures,
scientific-notation=true]{1040.674}) &  \num[round-precision=3,round-mode=figures,
scientific-notation=true]{15438.2415} (\num[round-precision=2,round-mode=figures,
scientific-notation=true]{814.19}) & \num[round-precision=3,round-mode=figures,
scientific-notation=true]{17487.13} (\num[round-precision=2,round-mode=figures,
scientific-notation=true]{1063.42}) & \num[round-precision=3,round-mode=figures,
scientific-notation=true]{37183.27} (\num[round-precision=2,round-mode=figures,
scientific-notation=true]{1463.75}) & \num[round-precision=3,round-mode=figures,
scientific-notation=true]{36637.717} (\num[round-precision=2,round-mode=figures,
scientific-notation=true]{1508.62}) & \num[round-precision=3,round-mode=figures,
scientific-notation=true]{95915.68} (\num[round-precision=2,round-mode=figures,
scientific-notation=true]{2617.64}) 
\\[1ex]
\hline
\\[-1ex]
\multirow{2}*{\textbf{Chicago Sketch}} & $\Delta = 5$ & \bfseries \num[round-precision=3,round-mode=figures,
scientific-notation=true]{33936.31} (\num[round-precision=2,round-mode=figures,
scientific-notation=true]{3104.20}) & \num[round-precision=3,round-mode=figures,
scientific-notation=true]{30907.83} (\num[round-precision=2,round-mode=figures,
scientific-notation=true]{2649.72}) & \num[round-precision=3,round-mode=figures,
scientific-notation=true]{35429.07} (\num[round-precision=2,round-mode=figures,
scientific-notation=true]{3617.66}) & \num[round-precision=3,round-mode=figures,
scientific-notation=true]{56402.82} (\num[round-precision=2,round-mode=figures,
scientific-notation=true]{6892.11}) & \num[round-precision=3,round-mode=figures,
scientific-notation=true]{61167.58} (\num[round-precision=2,round-mode=figures,
scientific-notation=true]{3541.55}) & \num[round-precision=3,round-mode=figures,
scientific-notation=true]{78675.82} (\num[round-precision=2,round-mode=figures,
scientific-notation=true]{8254.88})
\\[1ex]
(Medium) & $\Delta = 10$ & \bfseries \num[round-precision=3,round-mode=figures,
scientific-notation=true]{58755.56} (\num[round-precision=2,round-mode=figures,
scientific-notation=true]{8162.20})& \num[round-precision=3,round-mode=figures,
scientific-notation=true]{55603.73106} (\num[round-precision=2,round-mode=figures,
scientific-notation=true]{9568.82}) & \num[round-precision=3,round-mode=figures,
scientific-notation=true]{58897.63} (\num[round-precision=2,round-mode=figures,
scientific-notation=true]{8631.89}) & \num[round-precision=3,round-mode=figures,
scientific-notation=true]{112026.60} (\num[round-precision=2,round-mode=figures,
scientific-notation=true]{8200.41}) & \num[round-precision=3,round-mode=figures,
scientific-notation=true]{112053.30} (\num[round-precision=2,round-mode=figures,
scientific-notation=true]{9157.66}) & \num[round-precision=3,round-mode=figures,scientific-notation=true]{116305.90} (\num[round-precision=2,round-mode=figures,
scientific-notation=true]{9000.83}) 
\\[1ex]
\hline
\\[-1ex]
\multirow{2}*{\textbf{Sydney}} & $\Delta = 5$ & \bfseries \num[round-precision=3,round-mode=figures,
scientific-notation=true]{31369981.73} (\num[round-precision=2,round-mode=figures,
scientific-notation=true]{3942892.91}) & \num[round-precision=3,round-mode=figures,
scientific-notation=true]{26425831.67} (\num[round-precision=2,round-mode=figures,
scientific-notation=true]{2958763.23}) & \num[round-precision=3,round-mode=figures,
scientific-notation=true]{35032936.09} (\num[round-precision=2,round-mode=figures,
scientific-notation=true]{4080356.13}) & \num[round-precision=3,round-mode=figures,
scientific-notation=true]{79104369.51} (\num[round-precision=2,round-mode=figures,
scientific-notation=true]{6152217.34}) & \num[round-precision=3,round-mode=figures,
scientific-notation=true]{71883763.34} (\num[round-precision=2,round-mode=figures,
scientific-notation=true]{6251080.92}) & \num[round-precision=3,round-mode=figures,
scientific-notation=true]{125021411.68} (\num[round-precision=2,round-mode=figures,
scientific-notation=true]{8275923.10})
\\[1ex]
(Large) & $\Delta = 10$ & \bfseries \num[round-precision=3,round-mode=figures,
scientific-notation=true]{72605822.71} (\num[round-precision=2,round-mode=figures,
scientific-notation=true]{9583968.60}) & \num[round-precision=3,round-mode=figures,
scientific-notation=true]{64796139.25} (\num[round-precision=2,round-mode=figures,
scientific-notation=true]{9373361.90}) & \num[round-precision=3,round-mode=figures,
scientific-notation=true]{82429390.527} (\num[round-precision=2,round-mode=figures,
scientific-notation=true]{10777161.41}) & \num[round-precision=3,round-mode=figures,
scientific-notation=true]{177473823.31} (\num[round-precision=2,round-mode=figures,
scientific-notation=true]{20267510.59}) & \num[round-precision=3,round-mode=figures,
scientific-notation=true]{182584758.88} (\num[round-precision=2,round-mode=figures,
scientific-notation=true]{23994460.87}) & \num[round-precision=3,round-mode=figures,
scientific-notation=true]{227865777.98} (\num[round-precision=2,round-mode=figures,
scientific-notation=true]{33542552.55})  
\\[1ex]
\hline
\end{tabular}
{\\[1ex] \footnotesize $^*$CC is a special case of TASR where all demand is completely compliant.}
\end{table*}

\begin{figure*}[!t]
\vspace{-0.1in}
\centering
\begin{subfigure}[b]{0.495\textwidth}
\includegraphics[width=\textwidth]{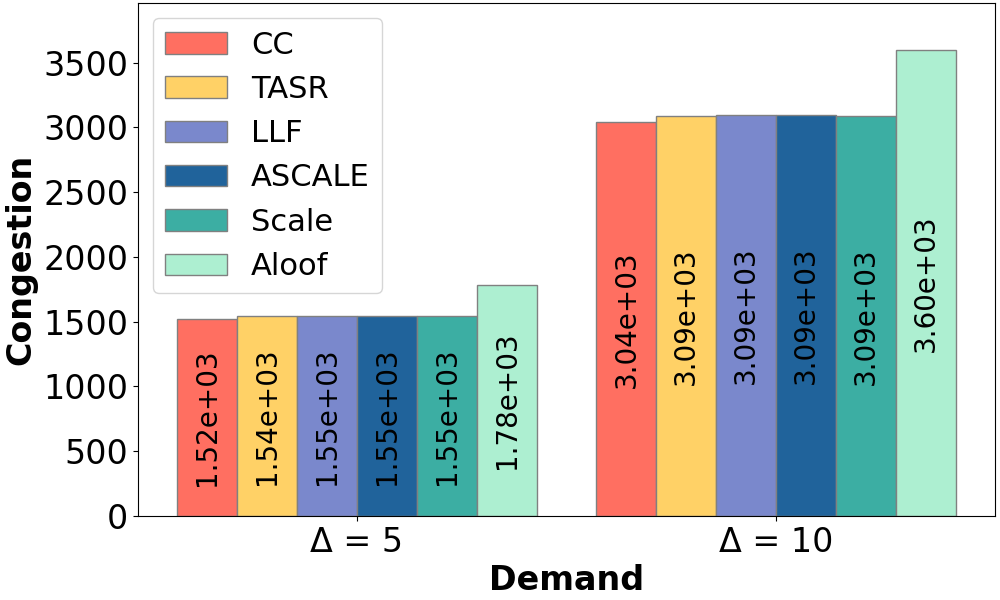}
\vspace{-0.25in}
\caption{Sioux Falls Network}
\label{Img: SF Subnet delta 5 and delta 10}
\end{subfigure}
\hfill
\begin{subfigure}[b]{0.495\textwidth}
\includegraphics[width=\textwidth]{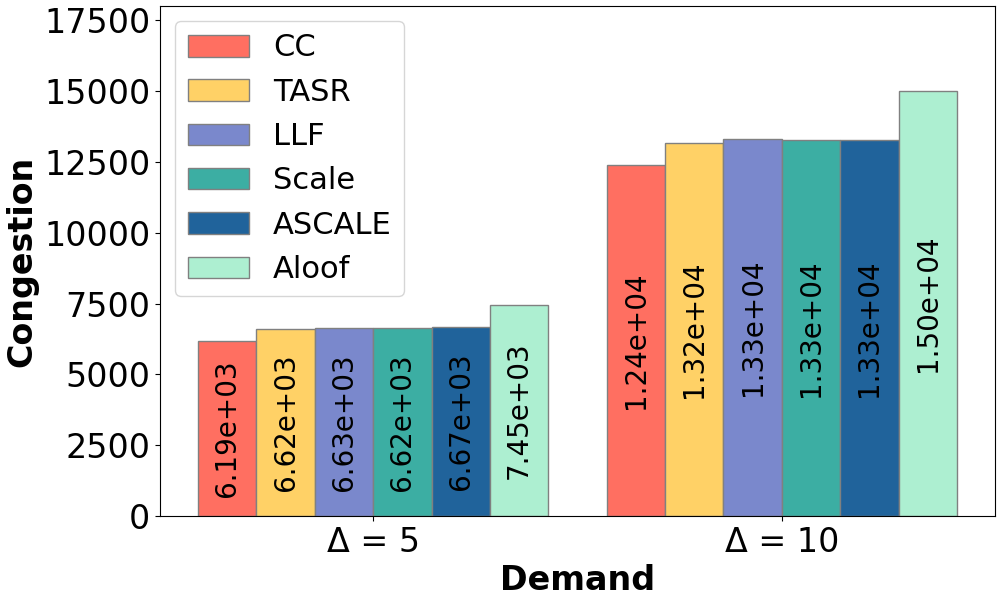}
\vspace{-0.25in}
\caption{Chicago Sketch Network}
\end{subfigure}
\vspace{-0.2in}
\caption{Comparison of Congestion for $\Delta = 5$ and $\Delta = 10$ on the Single-Commodity Sioux Falls and Chicago Sketch Networks.}
\label{Img: CS Subnet delta 5 and delta 10}
\vspace{-0.1in}
\end{figure*}

\begin{table}[!t]
\caption{Comparison of Efficiency Ratio in Single-Commodity Sioux Falls Network.}
\label{tab: SC Efficiency Ratios}
\centering
\begin{tabular}{c c c c c}
\hline
\\[-1ex]
\multicolumn{1}{c}{}{} & \multicolumn{4}{c}{\textbf{Demand Rate}} \\
\multicolumn{1}{c}{}{} & \multicolumn{2}{c}{\textbf{$\Delta = 5$}} & \multicolumn{2}{c}{\textbf{$\Delta = 10$}} \\
\multicolumn{1}{c}{}{} & \textbf{(SF)} & \textbf{(CS)} & \textbf{(SF)} & \textbf{(CS)} \\
\hline \hline
\\[-1ex]
\textbf{TASR} & \textbf{1.014986} & \textbf{1.069312} & \textbf{1.015465} & \textbf{1.064564} \\
\textbf{LLF} & 1.017447 & 1.070765 & 1.017162 & 1.074106 \\
\textbf{ASCALE} & 1.017598 & 1.077086 & 1.017412 & 1.072632 \\
\textbf{Scale} & 1.016645 & 1.069971 & 1.016307 & 1.073297 \\
\textbf{Aloof} & 1.173511 & 1.202740 & 1.184395 & 1.211893 \\
\hline
\end{tabular}
{\\[1ex] \footnotesize SF and CS are the Sioux Falls and Chicago Sketch networks respectively.}
\vspace{-3ex}
\end{table}

The second single-commodity network considered is a subgraph of the Chicago sketch network with a single origin-destination pair $(725, 700)$ with four available parallel paths comprised of 42 edges. Two scenarios are simulated: $\Delta = 5$ and $\Delta = 10$. In each scenario, one-sixth of the network demand is fully compliant and noncompliant respectively.

For the single-commodity Sioux Falls and Chicago Sketch networks, the average congestion resulting from each algorithm, as well as the standard deviation, is compared to a baseline of the complete compliance case of TASR, which we denote by CC. We consider values of $\Delta=5$ and $\Delta=10$ for both networks, and present the average runtime of each algorithm over 1,000 iterations. Specific details regarding each experiment run on the single-commodity networks will be provided in Section \ref{Sect: Results}.

\subsubsection{Multi-Commodity Network}
In order to study the scalability of TASR on a more realistic network in which overlapping routes apply (i.e., paths for one commodity may contain edges present in paths of a different commodity), simulation experiments are also conducted for the multi-commodity versions of the Sioux Falls, Chicago Sketch, and Sydney networks. Simulation experiments are run to compare the resulting network congestion and standard deviation of each algorithm on each network. In each simulation, a demand group is partitioned and assigned to a commodity at random in the network, and the Multi-Commodity TASR algorithm is run. Values of $\Delta = 5$ and $\Delta = 10$ have been considered for each network.

\section{Results and Discussions \label{Sect: Results}}
\subsection{Results for Single-Commodity Setting}
In the single-commodity Sioux Falls and Chicago Sketch networks with $\Delta = 5$ and $\Delta = 10$, the average network congestion as a result of TASR is compared to LLF, Scale, ASCALE, and Aloof, as well as to the baseline CC. Table \ref{tab: SF Single Commodity Travel Times} shows the resulting congestion for each algorithm in each scenario, as well as algorithmic runtimes. Note that TASR results in a clear improvement in network congestion compared as opposed to the alternative algorithms, since TASR is more considerate of probabilistic compliance of demand groups. These resuls are also presented graphically in Figures \ref{Img: SF Subnet delta 5 and delta 10} and \ref{Img: CS Subnet delta 5 and delta 10}, where the improvement in the performance of TASR as $\Delta$ increases is more apparent. Table \ref{tab: SC Efficiency Ratios} presents these results in the form of efficiency ratios between the resulting congestion of each algorithm compared to the special case of TASR with an entirely fully-compliant demand, CC, showing a clear performance advantage of TASR. The average travel time experienced by each unit of demand for each scenario considered in the single-commodity setting is shown in Table \ref{tab: SC Travel Time}. On average, there is an improvement in per-unit travel time of approximately $0.166\%$, $0.181\%$, and $0.145\%$ for $\Delta = 5$, $\Delta = 10$, and $\Delta = 15$ respectively in the single-commodity Sioux Falls network compared to the next-best performing algorithm, Scale.


Trust dynamics for repeated interactions with a homogenous demand with $\alpha = 0.5$ in the single-commodity Sioux Falls network with $\Delta = 10$ and $\Delta = 20$ are plotted in Figure \ref{Img: SF Subnet Repeated Trust 0.5 Delta 10 and 20}, with trust converging to $\alpha = 1.0$ for TASR after roughly 12 interactions with the system for $\Delta = 10$ and after roughly 23 iterations with $\Delta = 20$. While trust converges to full compliance after fewer interactions for LLF and Aloof for scenario in which $\Delta = 10$, and trust converges to full compliance after fewer interactions for LLF, Scale, ASCALE, and Aloof in the scenario in which $\Delta = 20$, this can be attributed to TASR producing lower overall network congestion with slightly higher per-demand travel times when compared to each other algorithm. In other words, there is a slight trade-off between trust and lower network congestion for the TASR algorithm. On the other hand, the average change in trust (over five hundred iterations) for individual independent interactions in the single-commodity Sioux Falls network is shown in Figure \ref{Img: SF Subnet Avg Trust} for each algorithm. Note that on average, the change in trust for a single iteration of each algorithm, the TASR algorithm produces comparable updated trust values to that of LLF, Scale, and ASCALE, and greater trust values than that of Aloof.

\subsection{Results for Multi-Commodity Setting}
The results obtained in the multi-commodity setting are very similar to those obtained in the single-commodity setting for networks with lower demand capacities. Table \ref{tab: MC Travel Time} depicts the congestion obtained as a result TASR compared to each algorithm and the complete compliance baseline in the Sioux Falls, Chicago Sketch, and Sydney networks for demands with $\Delta = 5$ and $\Delta = 10$. Similar to the single-commodity settings, TASR outperforms each comparable algorithm by a significant factor, which is more pronounced as the size of the network, network capacity, and amount of demand on the network scales.

\section{Conclusion
\label{Sect: Conclusion}}
In summary, a novel greedy Stackelberg routing framework was proposed to reduce the congestion of a traffic network in the presence of probabilistically compliant travelers (depending on trust in the system) in both single-commodity and multi-commodity traffic settings. The interaction between the system and diverse traveler demands was modeled as a Stackelberg game, and a novel trust-aware, greedy Stackelberg routing strategy was developed to send path recommendations to travel demands in order to mitigate network congestion. We demonstrated the performance gain of the proposed solution in terms of network congestion in comparison to LLF, Scale, ASCALE, and Aloof for networks of different demand capacities and topologies, and we presented the impact of each algorithm on the trust of various travel demands in single-interaction settings and repeated interaction settings. In the future we will investigate the performance of TASR on network congestion in a dynamic, repeated interaction settings across various networks and consider methods of maintaining the trust of traffic demands in order to ensure that the proposed algorithm can sustain its performance in a repeated setting.

\bibliographystyle{plain}
\bibliography{references}

\begin{thebibliography}{10}

\bibitem{beckmann1956studies}
Martin Beckmann, Charles~B McGuire, and Christopher~B Winsten.
\newblock Studies in the economics of transportation.
\newblock Technical report, 1956.

\bibitem{matteo2021}
Matteo Bettini.
\newblock Static traffic assignment using user equilibrium and system optimum -
  python code and network data.
\newblock
  \url{https://github.com/matteobettini/Traffic-Assignment-Frank-Wolfe-2021},
  2021.

\bibitem{bonifaci2007impact}
Vincenzo Bonifaci, Tobias Harks, and Guido Sch{\"a}fer.
\newblock The impact of stackelberg routing in general networks.
\newblock Technical report, Technical Report COGA Preprint 020-2007, TU Berlin,
  2007.

\bibitem{bonifaci2010stackelberg}
Vincenzo Bonifaci, Tobias Harks, and Guido Sch{\"a}fer.
\newblock Stackelberg routing in arbitrary networks.
\newblock {\em Mathematics of Operations Research}, 35(2):330--346, 2010.

\bibitem{colson2007overview}
Beno{\^\i}t Colson, Patrice Marcotte, and Gilles Savard.
\newblock An overview of bilevel optimization.
\newblock {\em Annals of operations research}, 153:235--256, 2007.

\bibitem{kaddoura2014optimal}
Ihab Kaddoura and Benjamin Kickh{\"o}fer.
\newblock Optimal road pricing: Towards an agent-based marginal social cost
  approach.
\newblock In {\em VSP working paper 14-01, TU Berlin, transport systems
  planning and transport telematics}. 2014.

\bibitem{kaporis2006price}
Alexis~C Kaporis and Paul~G Spirakis.
\newblock The price of optimum in stackelberg games on arbitrary single
  commodity networks and latency functions.
\newblock In {\em Proceedings of the eighteenth annual ACM symposium on
  Parallelism in algorithms and architectures}, pages 19--28, 2006.

\bibitem{karakostas2009stackelberg}
George Karakostas and Stavros~G Kolliopoulos.
\newblock Stackelberg strategies for selfish routing in general multicommodity
  networks.
\newblock {\em Algorithmica}, 53:132--153, 2009.

\bibitem{kolarich2022stackelberg}
Maxwell Kolarich and Negar Mehr.
\newblock Stackelberg routing of autonomous cars in mixed-autonomy traffic
  networks.
\newblock In {\em 2022 American Control Conference (ACC)}, pages 4654--4661.
  IEEE, 2022.

\bibitem{korilis1997achieving}
Yannis~A Korilis, Aurel~A Lazar, and Ariel Orda.
\newblock Achieving network optima using stackelberg routing strategies.
\newblock {\em IEEE/ACM transactions on networking}, 5(1):161--173, 1997.

\bibitem{krichene2016social}
Walid Krichene, Milena~Suarez Castillo, and Alexandre Bayen.
\newblock On social optimal routing under selfish learning.
\newblock {\em IEEE transactions on control of network systems}, 5(1):479--488,
  2016.

\bibitem{krichene2014stackelberg}
Walid Krichene, Jack~D Reilly, Saurabh Amin, and Alexandre~M Bayen.
\newblock Stackelberg routing on parallel networks with horizontal queues.
\newblock {\em IEEE Transactions on Automatic Control}, 59(3):714--727, 2014.

\bibitem{krichene2017stackelberg}
Walid Krichene, Jack~D Reilly, Saurabh Amin, Alexandre~M Bayen, T~Basar, and
  G~Zaccour.
\newblock Stackelberg routing on parallel transportation networks.
\newblock {\em Handbook of dynamic game theory}, pages 1--35, 2017.

\bibitem{macfarlane2019}
Jane Macfarlane.
\newblock Your navigation app is making traffic unmanageable.
\newblock {\em IEEE Spectrum}, 19, 2019.

\bibitem{roughgarden2001stackelberg}
Tim Roughgarden.
\newblock Stackelberg scheduling strategies.
\newblock In {\em Proceedings of the thirty-third annual ACM symposium on
  Theory of computing}, pages 104--113, 2001.

\bibitem{roughgarden2005}
Tim Roughgarden.
\newblock {\em Selfish routing and the price of anarchy}.
\newblock MIT press, 2005.

\bibitem{roughgarden2002bad}
Tim Roughgarden and {\'E}va Tardos.
\newblock How bad is selfish routing?
\newblock {\em Journal of the ACM (JACM)}, 49(2):236--259, 2002.

\bibitem{sharon2018traffic}
Guni Sharon, Michael Albert, Tarun Rambha, Stephen Boyles, and Peter Stone.
\newblock Traffic optimization for a mixture of self-interested and compliant
  agents.
\newblock In {\em Proceedings of the AAAI conference on artificial
  intelligence}, volume~32, 2018.

\bibitem{stabler2020}
B.~Stabler.
\newblock Transportation networks/chicago sketch.
\newblock
  \url{https://github.com/bstabler/TransportationNetworks/blob/master/Chicago-Sketch/ChicagoSketch_net.tntp},
  2020.

\bibitem{swamy2012effectiveness}
Chaitanya Swamy.
\newblock The effectiveness of stackelberg strategies and tolls for network
  congestion games.
\newblock {\em ACM Transactions on Algorithms (TALG)}, 8(4):1--19, 2012.

\bibitem{swamy2012}
Chaitanya Swamy.
\newblock The effectiveness of stackelberg strategies and tolls for network
  congestion games.
\newblock {\em ACM Transactions on Algorithms (TALG)}, 8(4):1--19, 2012.

\bibitem{wardrop1952road}
John~Glen Wardrop.
\newblock Road paper. some theoretical aspects of road traffic research.
\newblock {\em Proceedings of the institution of civil engineers},
  1(3):325--362, 1952.

\end{thebibliography}

\end{document}